\RequirePackage{fix-cm}
\documentclass[smallextended]{svjour3}       
\smartqed  
\usepackage{graphicx}

\usepackage{macros}
\journalname{Journal of Applied and Computational Topology}
\begin{document}

\title{
    Approximate Nearest Neighbors \\
    in the Space of Persistence Diagrams\thanks{
        BTF is supported by
            NSF CCF 1618605,
            NSF ABI 1661530, and
            NIH/NSF DMS 1664858.
        XH is supported by China Scholarship Council under program
            201706240214 and
            by the Fundamental Research Funds for the
            Central Universities under Project 2012017yjsy219.
        ZL is supported by a Shandong Government Scholarship.
        SM is supported by NSF CCF 1618605.
        DLM is supported by NSF ABI 1661530.
        BZ is partially supported by NSF of China under project 61628207.
    }
}


\author{
    Brittany Terese Fasy \and
    Xiaozhou He \and
    Zhihui Liu \and
    Samuel Micka \and \\
    David L.\ Millman \and
    Binhai Zhu
}

\authorrunning{B.T.\ Fasy, X.\ He, Z.\ Liu, S.\ Micka, D.L.\ Millman, and B.\ Zhu}

\institute{
    B.T. Fasy \at
        School of Computing and Dept.\ of Mathematical Sciences,
        Montana State University,
        Bozeman, MT, USA
        \email{brittany.fasy@montana.edu} \and
    X. He \at
        Business School,
        Sichuan University,
        Chengdu, Sichuan, China;
        and School of Computing, Montana State University,
        Bozeman, MT, USA
        \email{xiaozhouhe126@qq.com} \and
    Z. Liu \at
        School of Computer Science and Technology,
        Shandong Technology and
        Business University,
        Yantai, Shandong, China; and
        School of Computing,
        Montana State University, Bozeman, MT, USA
        \email{dane.zhihui.liu@gmail.com} \and
    S. Micka \at
        School of Computing,
        Montana State University, Bozeman, MT, USA
        \email{samuel.micka@msu.montana.edu} \and
    D. L.\ Millman \at
        School of Computing, Montana
        State University, Bozeman, MT, USA
        \email{david.millman@montana.edu} \and
    B. Zhu
        School of Computing,
        Montana State University, Bozeman, MT, USA
        \email{bhz@montana.edu}
}

\date{Received: date / Accepted: date}

\maketitle

\begin{abstract}
Persistence diagrams are important tools in the field of topological data analysis
that describe the magnitude of features in a filtered topological
space. However, we show that the doubling dimension of $(M,m)$-bounded
persistence diagrams is infinite and, as a result,
current approaches for comparing a persistence diagram to a
set of other persistence diagrams is linear in the number of diagrams or do
not offer performance guarantees.
In this paper, we provide the first approach supporting
approximate nearest neighbor search in the space of persistence diagrams
using the bottleneck distance.
Given a set~$\Gamma$ of $n$ $(M,m)$-bounded persistence diagrams,
each with at most $m$ points, we snap-round the points of each diagram to
specific points on a uniform grid and produce a key for each possible snap-rounding.
Then, we propose a data structure with~$\tau$ levels $\bigds$ storing
all snap-roundings of each persistence diagram in~$\Gamma$ at each resolution.
This data structure has size
$\spcone$ to account for varying grid resolutions, snap-roundings, and
the deletion of points with low persistence.
To search for a persistence diagram, we compute a key for a query
diagram by snapping each point to a grid point and deleting points of low
persistence.
Furthermore, as the grid parameter decreases,
searching our data structure yields a six-approximation of the nearest
diagram in $\Gamma$ in~$\qtone$~time and twenty-four
approximation of the $k$th nearest diagram in $\qtknn$ time.

\end{abstract}

\section{Introduction}
Computational topology is a field at the intersection of mathematics (algebraic
topology) and computer science (algorithms and computational geometry).
In recent years, the use of techniques from computational topology in application
domains has been on the rise~\cite{ahmed2014local, giusti2015clique, le2015simplicial}.
Furthermore, persistence diagrams can be used
to reconstruct different types of simplicial complexes, which can be used to represent
geometric objects and point clouds~\cite{turner2014persistent, belton2018learning}.
These results provide new avenues to explore object classification
and recognition in new and enlightening ways.
More generally, current research is applying techniques from computational
topology to big data.
Computing the distance from a query diagram to a set of persistence diagrams
using a linear number of computations, however,
can be computationally expensive since computing the bottleneck distance
between two diagrams requires $O(m^{1.5}\log m)$ time for
$m$ points \cite{Kerber16}.
To address the expense, preliminary work by Kerber and
Nigmetov~\cite{mk2018spanners} looked at understanding the space of persistence
diagrams through building a cover tree of a set of diagrams but are not able to
offer worst-case performance guarantees due to the doubling dimension.
To reduce the complexity of comparing a query diagram to a set of
diagrams, Fabio and Ferri represented persistence diagrams as
complex polynomials and compared the persistence diagrams using
complex vectors storing coefficients for the polynomials \cite{di2015comparing}.
The research in \cite{di2015comparing}, however, is experimental and offers no
performance guarantees on the distance between two diagrams
deemed to be close to one another by comparing the complex vectors.
The idea of converting persistence diagrams to complex vectors was
extended by Wang et al. in 2019 \cite{wang2019polynomial}. Wang et al.
prove stability results of the vector representation. Specifically,
for two diagrams~$D_1$ and $D_2$, the distance between the
resulting vectors $v_{D_1}$
and $v_{D_2}$ has the following upper-bound:
$$
	||v_{D_1} - v_{D_2}||_1 \le \sqrt{2}\large( 1 + \frac{\sqrt{\pi}}{\sigma}
	\large)d_{W,1}(D_1, D_2),
$$
where $d_{W,1}$ denotes the $1$-Wasserstein distance and $\sigma$ is a
variance parameter~\cite[Theorem 1]{wang2019polynomial}.
However, without a lower-bound on the distance, this result is not
applicable for guaranteeing the distance between a returned diagram and
the nearest neighbor.
We address a the problem of near-neighbor searching,
answering near neighbor queries in the space of persistence diagrams,
and providing a means of querying for near diagrams with
performance~guarantees.

Nearest neighbor search is a fundamental problem in computer science, i.e.,
databases, data mining and information retrieval, etc. The problem was
posed in 1969 by Minsky and Papert \cite{MP69}. For data in low-dimensional
space, the problem is well-solved by first computing a search structure
on the Voronoi diagram of the data points and then performing point
location queries for query~points \cite{Ed87}.
When the dimension is large, such a method is known to be
impractical as the query time typically has a constant factor which is
exponential in dimension (known as the ``curse of dimensionality'')~\cite{Cl94}.
Then,
researchers resort to approximate nearest neighbor search~\cite{arya1998ann}.

In many applications, for approximate nearest neighbor queries, the data in
consideration are not necessarily (high-dimensional) points in Euclidean space. In 2002,
Indyk considered the data to be a set of $n$ polygonal curves (each with
at most $m$ vertices) and the distance between two curves is the discrete
Fr\'echet distance. In particular, a data structure of
exponential size was built so that
an approximate nearest neighbor query (with a factor $O(\log m+\log\log n)$)
can be done in $O(m^{O(1)}\log n)$ time~\cite{Indyk02}. Most recently,
Driemel and Silvestri used locality-sensitive hashing to answer near neighbor
queries (within a constant factor) in $O(2^{4md}m\log n)$ time using $O(2^{4md} n \log n + nm)$ space (this bound is practical only for some $m=O(\log n)$) \cite{DS17}.
In the case of persistence diagrams under the bottleneck distance, no results
exist on finding the nearest neighbor or even a near neighbor efficiently
while offering performance guarantees without computing pairwise
bottleneck distances.

\paragraph*{Our Contributions}
We study the near neighbor search and the $k$-near neighbor problems
in the space of persistence diagrams under the bottleneck distance.
As observed by Kerber and Nigmetov~\cite{kerber2019metric},
the significant research in approximate and exact nearest neighbor
searching~\cite{Indyk98, HIM12, Indyk02, arya1998ann, andoni2015optimal,
derryberry2008achieving,chan2002closest,beygelzimer2006covertree,mk2018spanners}
does not carry over to the space of persistence
diagrams with the bottleneck metric.
Specifically, Kerber and Nigmetov observed in~\cite{kerber2019metric}
and \cite{mk2018spanners} (and
in~\thmref{doubling} we show for a more restricted setting) that the space of
persistence diagrams with the bottleneck metric has infinite doubling
dimension.
As discussed by Clarkson's~\cite{clarkson2006nnmetric}(page 30),
general divide-and-conquer algorithms over metric spaces have a running time
that is exponential in the doubling dimension.
As such, \thmref{doubling} implies that
searching approaches that decompose space, such as cover trees or quad trees,
may have nodes with an infinite number of children and are not practical.
Kerber and Nigmetov tested cover trees on sets persistence diagrams
clustered around sufficiently spaced centers and data from the McGill
shape benchmark\footnote{\url{http://www.cim.mcgill.ca/~shape/benchMark/}}
and found that cover trees reduced the
number of computations but still acknowledge the inability to
offer worst-case performance guarantees \cite{mk2018spanners}.

This work also draws similarities to hashing and searching with bottleneck
distance for multi-point
datasets in two-dimensions~\cite{heffernan1994approximate,
efrat1996llmprovements,indyk2003approximate,akutsu1998determining,
benkert2012approximate}.
However, most existing approaches require that the points sets have the same
cardinality, an assumption not necessary when computing
the bottleneck distance between persistence diagrams.
For persistence diagrams, equal cardinality comes from the diagonal having
infinite multiplicity, which allows off diagonal points to
be mapped to the diagonal.
While some approaches, such as \cite{indyk2003approximate,yon2017finding},
overcome the equal cardinality assumption,
the output of their approaches is not necessarily a bottleneck matching
and not applicable when comparing persistence diagrams.
As such, the crux of this research is managing matchings of points with the
diagonal.

Up until now, the fastest known approach with performance guarantees
for searching in the space of persistence diagrams was to compute the
bottleneck distance from the query diagram to all diagrams in the data set using
the methods introduced by Kerber et al.~\cite{Kerber16}.
In this paper, we present the first efficient solution to
searching in the space of persistence diagrams
with performance guarantees.
We summarize our results as follows. Given a set~$\Gamma$ of~$n$
$(M,m)$-bounded persistence diagrams, we propose a key
function that produces a set of keys for each of the $O(5^m)$
snap-roundings of each diagram in $\Gamma$ to grid points.
The keys are stored in a data structure~$\bigds$, with $\tau$ levels
using~$\spcone$ space.
This data structure is capable of answering
queries of the form: given a query persistence diagram $Q$ with
at most~$m$ points, return a six-approximation of the nearest neighbor
to $Q$ in $\Gamma$ in $\qtone$ time or return a set of $k$ diagrams, each of which
being a twenty-four factor approximation of the $k$th nearest diagram,
in $\qtknn$ time.

\section{Preliminaries}
In this section, we give necessary definitions for persistence diagrams,
bottleneck distance and additional concepts used throughout
this paper. We assume that the
readers are familiar with the basics of algorithms \cite{CLRS01}.

\subsection{Persistence Diagrams and Bottleneck Distance}\label{ss:persistence}

\emph{Homology} is a tool from algebraic topology that describes the so-called
\emph{holes} in topological spaces by assigning the space an abelian group for
each dimension.
When we are not given an exact topological space, but an estimate of it, we need
to introduce some notion of scale.  If each scale parameter~$\tau$ is assigned a
topological space $X_{\tau}$ such that $X_{\tau}$ changes nicely with~$\tau$,
then we track these changes using \emph{persistent homology}.
For further details on classical homology theory, the readers are referred
to \cite{ATH,munkres:at} for homology and \cite{EH10} for persistent homology.
In this paper, we are working in the \emph{space of persistence diagrams} under
the bottleneck distance (which
we will make more precise next), and are not concerned with where these diagrams
came from.

Persistent homology tracks the \emph{birth} and
\emph{death} of the
topological features (i.e., the connected components, tunnels, and
higher-dimensional `holes') at multiple scales.
A {\em persistence diagram} summarizes this information by representing the birth
and death times ($b$ and~$d$, respectively) of homology generators as points  $(b,d)$ in the extended
plane.
These birth and death times correspond to the appearances and merging
of topological features in a filtered topological space and, as a result, are
not necessarily positive.
We comment that we could also represent a persistence diagram as a set of
half-open intervals
(barcodes) in the form of $[b,d)$ as
in~\cite{zomorodian2005computing,chazal2016structure}.
We focus on the former representation and lay a grid over the diagram (see
\ssecref{lattices}) and consider the $L_\infty$ distance.
Let~$\diag$ denote the diagonal (the line $y=x$)
with points on~$\diag$ represented with infinite multiplicity.
Notice that points
with small persistence are close to the diagonal
and points with high persistence are far from the diagonal;
in particular, the point~$(b,d)$ has
distance~$\frac{1}{2}|d-b|_\infty$ from $D$.
In what follows, we set $M,m >0$ and consider diagrams with at
most $m$ off-diagonal points such that each
off-diagonal point $(a,b)$ satisfies $|a| \leq M$ and
$|b| \leq M$ or $|b| = \infty$.
We call such persistence
diagrams $(M,m)$-bounded; see \figref{diagram}.

\begin{figure}[htbp]
\begin{center}
\includegraphics[width=0.40\textwidth]{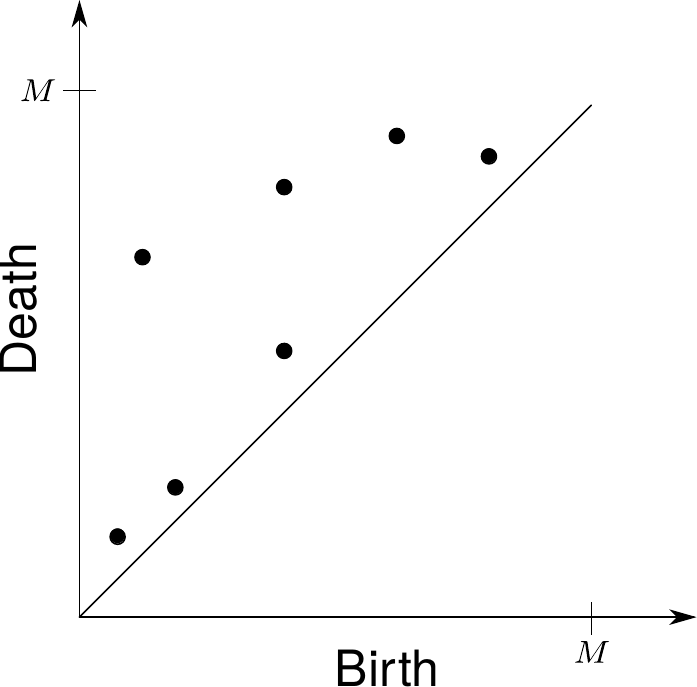}
\end{center}
    \caption[Example of an $(M,7)$-bounded persistence diagram]
    {An example of an $(M,7)$-bounded persistence diagram. Each point $p=(b,d) \in P
    \backslash D$
represents a topological feature--in particular, a homology generator. A point
    that
is close to the diagonal (i.e., has small persistence), cannot be easily
    distinguished from \emph{topological noise}.
Standard persistence has points above the line $y=x$; however, extended
persistence allows points above or below the diagonal.
}\label{fig:diagram}
\end{figure}

Given persistence diagrams $P$ and $Q$,
the {\em bottleneck distance} between them~is:
\[d_B(P,Q)=\inf_{\phi} \sup_{p\in P}\|p-\phi(p)\|_{\infty},\]
where the infimum
is taken over all bijections $\phi:P\rightarrow Q$.
Notice that such an infimum exists, since $||\cdot||_{\infty}$ is nonnegative and
there exists at least one bijection~$\phi$ with finite bottleneck distance
(namely, the one that matches every~$p$ in $P \backslash \diag$ to $\phi(p)$ and
every $d$ in~$\diag$ gets matched to
itself).
Let $\proj \colon \R^2 \to \diag$ be the orthogonal projection of a point~$p \in
\R^2$ to the closest point on $\diag$.
We define a {\em matching} between the points of~$P$ and $Q$
to be a set of edges such that no point in $P$ or $Q$ appears more than once.
We interpret these edges as pairing a point $p \in P$ with either off-diagonal point $q\in Q$
or $\proj(q)$, and a point $q\in Q$ with either points an off-diagonal point $p\in P$ or $\proj(q)$.
Furthermore, a matching is {\em perfect}
if every $p\in P$ and $q\in Q$ is matched,
i.e., every point is paired with the diagonal or a point from the other diagram and
every point has degree one, see \figref{matching} for an example.
Letting $P_0$ and $Q_0$ be the sets of off-diagonal points in $P$ and $Q$,
respectively.  Then, if $d_B(P,Q)= \varepsilon$, a perfect matching $M$ between $P'=P_0 \cup
\proj(Q_0)$ and $Q' = Q_0 \cup \proj(P_0)$ exists such that the length of each
edge in $M$ is at most $\varepsilon$; again, see \figref{matching}. In this
light, computing the
bottleneck distance is equivalent to finding a perfect matching between~$P'$ and~$Q'$ that
minimizes the length of the longest edge; see~\cite[\S VIII.4]{EH10}
and~\cite{Kerber16}, which use results from graph matching~\cite{lovasz2009matching,hungarian,EIK01,HK73}.
The space of persistence diagrams under the bottleneck distance is, in fact, a
metric space. Throughout this paper,
we use $\dgms$ to denote the (metric) space of $(M,m)$-bounded persistence
diagrams under the bottleneck distance.

\begin{figure}[htbp]
    \begin{center}
        \includegraphics[width=0.40\textwidth]{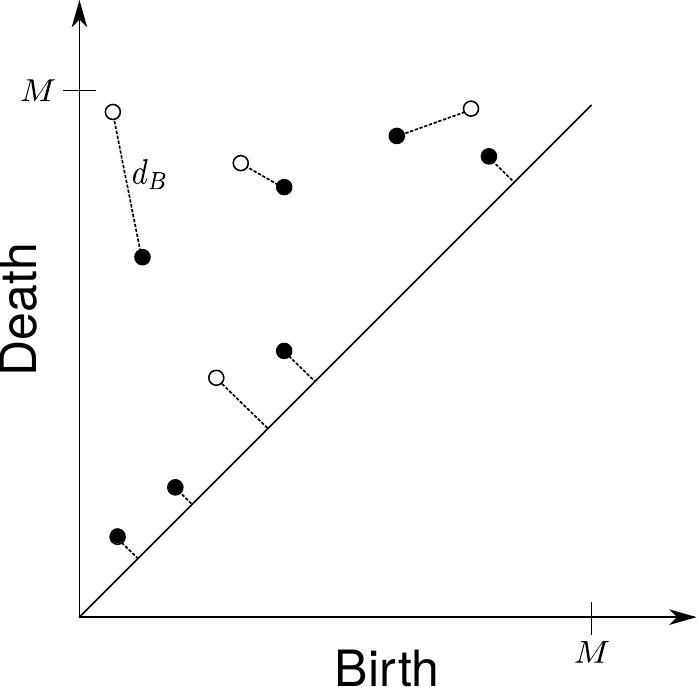}
    \end{center}
    \caption[Perfect matching between two persistence diagrams]
    {An example of a perfect matching between diagrams $P$ (the solid black points) and
    $Q$ (the circles).  Notice that $P$ and $Q$ have a different number of
    off-diagonal points.  However, since $\diag \subset P$ and $\diag \subset
    Q$, points in either diagram may match with the diagonal. A key insight in the algorithm for
    computing $d_B$ is that if a point $p$ is matched to a point on~$D$, then we
    can simply match $p$ with $\proj(p)$.
    }\label{fig:matching}
\end{figure}

Next, we prove a property about the space $\dgms$ with the bottleneck
metric known as the \emph{doubling dimension}.
Following Clarkson's definition of the doubling space \cite{clarkson2006nnmetric}:

\begin{definition}[Doubling Dimension]
    Let $X$ be a metric space with metric~$d$.
    The space $X$ is said to be \emph{doubling} if there exists
    a constant integer~$C>0$ such that for every $x \in X$ and
    any $r>0$ the open ball $B(x,r)$ can be covered by at most $2^C$
    balls of radius $\frac{r}{2}$.
    The $\emph{doubling constant}$ of $X$ is $2^C$
    and the $\emph{doubling dimension}$ of $X$ is $C$.
\end{definition}

Then, the following theorem highlights the problem
with employing search approaches dependent on decomposing the search space,
such as cover trees. We note that similar observations were
made about persistence diagrams (not~$(M,m)$-bounded)
under the bottleneck distance in
\cite{mk2018spanners, kerber2019metric}.

\begin{theorem}\label{thm:doubling}
    The space of $(M,m)$-bounded persistence diagrams
    under the bottleneck distance has infinite doubling dimension.
\end{theorem}
\begin{proof}
    Assume, for contradiction, that there exists a constant $C$ such
    that for any positive radius $r$ and all $P\in \dgms$ that $B(P,r)$ can
    be covered by $2^C$ balls of radius~$\frac{r}{2}$.
    Let~$r = \frac{M}{2^Cm}$ and let $P\in \dgms$ have exactly
    one point $\frac{r}{2}$ from the diagonal.
    Let~$\mathbb{B}$ be the region greater than or equal to $\frac{r}{2}$
    but strictly less than $r$ from the diagonal.
    Note that~$\mathbb{B}$ can be decomposed into~$\frac{2M}{r}$ disjoint boxes
    with edge length $\frac{r}{2}$, we denote this decomposition as~$\mathbb{B}'$.
    Let~$\mathbb{D}$ be the set of diagrams with $m$ points,
    where each of the $m$ points lie at the center of a different box
    in~$\mathbb{B}'$.
    Each~$Q\in \mathbb{D}$ has $m$ points lying at $m$ of
    $\frac{2M}{r} = 2^{C+1}m$ possible
    locations. For any~$Q,Q' \in \mathbb{D}$, where $Q\neq Q'$ we observe
    that~$d_B(P,Q) < r$, $d_B(P,Q') < r$, and $d_B(Q,Q') \ge \frac{r}{2}$.
    In other words, we must cover every diagram in $\mathbb{D}$ but
    each requires a separate open ball of radius~$\frac{r}{2}$. However,
    \[2^C < 2^{C+1} \le  {2^{C+1}m \choose m} = {\frac{2M}{r} \choose m} =
    |\mathbb{D}|,\]
    a contradiction to the claim that $B(P,r)$ can be covered
    by $2^C$ balls of radius~$\frac{r}{2}$.
\end{proof}

\subsection{Uniform Grid}\label{ss:lattices}
A $M$ bounded \emph{uniform grid} in~$\R^2$ is a Cartesian grid with
$\eta ^2$ grid
cells all contained in the box $[-M, M]^2$ and is denoted~$\lattice{M}{\eta}$.
We include $-M$ in the bounding box because the coordinates of points
in the persistence diagrams are not necessarily positive.
Grid cells are defined by four grid edges and four grid points.
The \emph{grid parameter}, denoted~$\delta = 2M/\eta$,  is the cell width of each grid cell.
Since we are working with persistence diagrams, which may have points at
$\infty$, we include two additional one-dimensional grids to handle points with
a single infinite coordinate and one additional zero-dimensional grid
cell to handle points with two infinite coordinates.
The \emph{complexity} of the grid is the number of grid cells:
$|\lattice{M}{\eta}|=\eta^2 + 2\eta + 1$.
For simplicity of exposition, we assume that no input points lie on either grid
edges or equidistant from any two grid points. Thus, nearest
grid points in the grid are unique, and every point has exactly four grid
points defining the grid cell containing it.

\section{Generating and Searching Persistence Keys}
\label{sec:hashing}
In this section, we define a key function that maps a persistence diagram in $\dgms$ to a
vector in~$\Zpos{a}$ where the exponent $a$ is a function of the
number of points in the grid.
Hence, as $a$ increases, the keys become more discerning.
We order the diagrams using the
dictionary order on $\Zpos{a}$, and store the keys in a multilevel
data structure that supports binary search.
We note here that the hierarchical grid is adapted from
approaches to locality-sensitive hashing~\cite{Indyk98,HIM12,heffernan1994approximate}.
More recent general results on locality-sensitive can be found in~\cite{andoni2015practical,andoni2015optimal,christiani2017framework, anagnostopoulos2018algorithms, andoni2015tight}.

Let $M,m >0$ and $\eta \in \Z_+$.
Let $P\in \dgms$ be a persistence diagram.
 We consider the grid~$\lattice{M}{\eta}$.
We then snap each off-diagonal point $p \in \nodiag{P}$ to a grid point
$\rho_i \in \lattice{M}{\eta}$
and count the multiplicity $\pi_i$ for each grid point.
The number of grid points from two-dimensional grid cells
is $(\eta+1)^2$, the number of grid points from one-dimensional grid
cells is $2(\eta+1)$, and there is a single zero dimensional grid cell
with one point.
We note that while our key function was inspired by the hash function
of~\cite{DS17} that ignored multiplicities, we must count the multiplicity of
duplicated grid~points.

Recall from \ssref{lattices} that the grid
parameter is $\delta=2M/\eta$.  We define our key function~$\keyfcn \colon
\dgms \times \mathbb{G} \to \Z^a$ by:
\[ \key{P}{\lattice{M}{\eta}} = \sum_{p \in \nodiag{P}} e_{nn(p)},\]
where $\mathbb{G}$ denotes the set of all uniform grids, $nn(p)$ maps each
off-diagonal~$p \in P$ to the index of the nearest grid point and
$e_i$ is the $i^{th}$ standard basis vector in $\Z^{a}$, where
$a=(\eta+1)^2 + 2(\eta + 1) + 1$; see \figref{snapping} for an example.
For simplicity of the proofs to follow and so $\keyfcn$ is well-defined,
we assume that no persistence point lies on a grid edge of $\lattice{M}{\eta}$
or equidistant to any two grid points.
Of course, since many coordinates of
$\key{\cdot}{\cdot}$ are zero, we store it using a sparse vector representation; moreover,
for the empty diagram~$D$ (i.e., with no point but the line $y=x$), we notice
that~$\key{D}{\cdot}=0$.
\begin{figure}[tb]
\begin{center}
\includegraphics[width=0.33\textwidth]{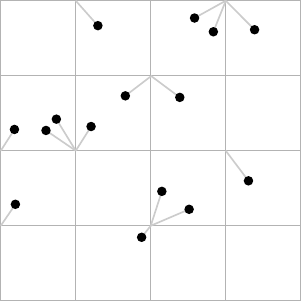}
\end{center}
\caption[Key function and snap-rounding points to a grid]
{The key function snap-rounding each point of the input set to the
nearest grid point. Note that the actual rounding produced by the key
function is denoted as a line from each point to the corresponding grid
point that it is rounded to.
}\label{fig:snapping}
\end{figure}

\begin{remark}
    This vector could also correspond to a product of  prime numbers, where
$\{\sigma_j\}$ is an ordered set of $a$ prime numbers.  Then, we have a unique
integer $\prod_j \sigma_j^{v_j}$ for each vector $v=(v_1, v_2, \ldots, v_{a})
\in \Z^{a}$.
Doing so would put us in the more conventional setting, where we
have indices into a hash table instead of keys.  However, we would then either need to have a
pre-generated list of $a$ primes (which adds to our storage space) or must
account for computing the primes (which adds time complexity).
\end{remark}

Suppose we snap-round $P$ before applying $\keyfcn$; a natural choice for
rounding each $p \in P$ would be to one of the four grid points defining
the grid cell containing $p$.  Formally:

\begin{definition}[Snap Sets and Canonical Ordering]\label{def:ordering} Given
    a grid $\nLattice$ with grid parameter~$\delta$, let
    $\snap{P}{\nLattice}$ denote the set of all possible snap-roundings of
    $P$ obtained by allowing each~$p \in P$ to snap-round to one of the grid
    points within $L_\infty$ distance $\delta$ of $p$, i.e., one of the
    grid points bounding the cell containing $p$.
    For example, points with finite coordinates lie within $\delta$ of
    four grid points, points with one infinite coordinate lie within
    $\delta$ of two grid points, and points with two
    infinite coordinates lie within $\delta$ of one grid point.
    Let $\delsnap{P}{\nLattice}$ denote
    the set of all snap-roundings of $P$ obtained by additionally allowing
    $p \in P$
    distance less than or equal to $\delta$ from the diagonal to be optionally deleted; see
    \figref{noisydgm} for an example of points that are eligible for removal.
\end{definition}

\begin{figure}[ht]
    \begin{center}
        \includegraphics[width=0.40\textwidth]{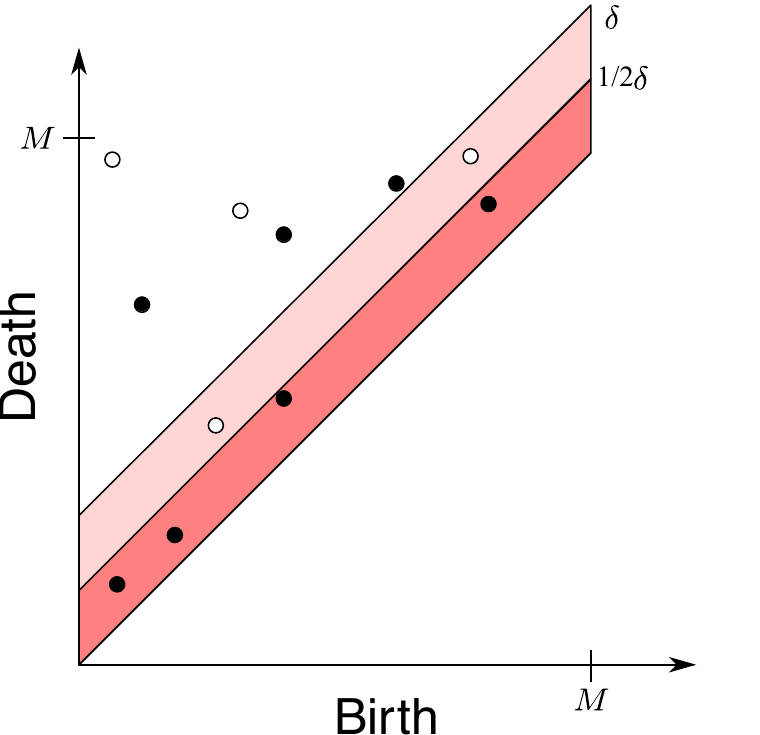}
    \end{center}
    \caption[Thresholds for point deletion near diagonal]
    {
        An example demonstrating the thresholds where points near the
        diagonal may be deleted. We consider diagrams $P,Q \in \dgms$
        denoted with black and white points, respectively. The two thresholds
        are shown as bands above the diagonal denoting distances
        $\frac{1}{2}\delta$ and
        $\delta$ from the diagonal. Removing all black points within
        the first threshold will produce the diagram~$\widetilde{Q}$ from $Q$.
        Removing some subset of white points within the second threshold
        may produce~$P_*$ from $P$.
    }
    \label{fig:noisydgm}
\end{figure}

For the remainder of this
paper, without loss of generality, we only
consider points with four snap-roundings since points infinite coordinate
values follow the same argument as points with a finite death time that
are not too ``near the diagonal'', but with fewer snap-roundings.
Each point~$p \in P$, not lying on a grid edge, can snap to the four
grid points defining the grid cell containing~$p$. Then, we bound the
number of snap-roundings for a fixed diagram:

\begin{lemma}[Enumerating Keys]\label{lem:enumerating}
    If $P \in \dgms$ and if $\nLattice$ is a uniform grid centered on
    $[-M,M]^2$, then
    the number of keys in $\snap{P}{\nLattice}$ and
    $\delsnap{P}{\nLattice}$ have the following upper bounds:
    \begin{itemize}
        \item $\snap{P}{\nLattice}$ has size~$O(2^{2m}) = O(4^m)$.
        \item $\delsnap{P}{\nLattice}$ has size~$O((2^2+1)^m) = O(5^m)$.
    \end{itemize}
\end{lemma}


We are now almost ready to prove \thmref{collision}, which shows that a query
diagram $Q$ collides with a snapping of diagram $P$ if and only if $P$ and $Q$
are close diagrams.  (Note that this `if and only if' statement uses
asymmetric notions of \emph{close}).  We first prove a
simplified version in \lemref{boundedcollision}, where we consider the perfect
matching problem in the extended plane $\Rbar^2$.  In this case, the proof is
made
easier as the two diagrams necessarily have the same number of points (as
otherwise, a perfect matching is not possible).  Then, to prove
\thmref{collision},
we delete points that are less than or equal to distance
$\frac{1}{2}\delta$ of $D$ in the query
diagram $Q$ and observe
that some of these points could have been matched with off-diagonal points in
$P$ with persistence up to, and including,~$\frac{3}{2}\delta$.
We note that, for the sake of space, we include all proofs in the appendix.

\begin{lemma}[Collision without Diagonal
Interference]\label{lem:boundedcollision}
    Let $P,Q \subset \Rbar^2$ be finite $(M,m)$-bounded point clouds.
    Let $\eta \in \Z_+$ and consider the
grid~$\nLattice=\lattice{M}{2\eta}$.
    Let $\delta$ denote the grid parameter.
    Then,

    \begin{enumerate}
        \item if
    $ \exists P_* \in \snap{P}{\nLattice}$
such that
            $\key{P_*}{\nLattice}=\key{Q}{\nLattice}$, then
            $d_B(P,Q) \leq \frac{3}{2} \delta$;
        \label{lemPart:snapImpliesBound}
        \item if $d_B(P,Q) \le \frac{1}{2} \delta$, then
    $\exists P_* \in \snap{P}{\nLattice}$ such that
            $\key{P_*}{\nLattice}=\key{Q}{\nLattice}$.
        \label{lemPart:boundImpliesSnap}
    \end{enumerate}
\end{lemma}

The above lemma is restricted to matchings of points in $\Rbar^2$.
Next, we generalize \lemref{boundedcollision},
by allowing matchings to the diagonal. This is the
central theorem of this~paper.

\begin{theorem}[Collisions between Diagrams]\label{thm:collision}
    Let $P,Q \in \dgms$.
    Let~$\eta \in \Z_+$ and consider the grid~$\nLattice =
\lattice{M}{2\eta}$.
    Let $\delta$ denote the grid parameter,
    and let $\widetilde{Q}$ be the diagram obtained from $Q$ by removing
            all points less than or equal to distance $\frac{1}{2}\delta$ from the diagonal.~Then:
            \begin{enumerate}
                \item If $ \exists P_* \in
\delsnap{P}{\nLattice}$
                    such that
                    $\key{P_*}{\nLattice}=\key{\widetilde{Q}}{\nLattice}$, then
                    $d_B(P,Q) \leq \frac{3}{2} \delta$.
                \label{thmPart:snapImpliesBound}
                \item If $d_B(P,Q) \le \frac{1}{2} \delta$, then
$\exists P_* \in \delsnap{P}{\nLattice}$
                    such that
                    $\key{P_*}{\nLattice}=\key{
                        \widetilde{Q}}{\nLattice}$.
                \label{thmPart:boundImpliesSnap}
            \end{enumerate}
\end{theorem}

This result implies that diagrams with a small bottleneck distance relative
to the chosen~$\delta$ value
will have matching keys generated by the hashing function
while diagrams with a large bottleneck distance, relative to~$\delta$, will not.
Next, using \thmref{collision}, we discuss a multi-level data structure that,
for some query diagram $Q$, supports searching for approximate nearest neighbors in
$\dgms$.

\section{Finding Approximate Nearest Neighbors}\label{sec:approx}

In \thmref{collision}, we saw that for a query diagram $Q$ and grid
parameter~$\delta$,
$Q$ will share a key with some diagram~$P$ `if and only if' they are close, with
respect to the chosen scale $\delta$.  To find the near-neighbor, we must select a $\delta$ with
the correct relationship to $d_B(P, Q)$.  The relationship  presents two
problems. First, how do we determine the correct value for $\delta$? Second, a
single~$\delta$ value would rarely be sufficient for all queries.

In this section, we build a multi-level data structure to support approximate
nearest neighbor queries in the space of persistence diagrams.
Each level of the data structure corresponds to a grid with a different resolution.
In the previous
section, we needed a flexible notion for $\snapName$ and $\delsnapName$, but in
this section, the data structure level and grid are dependent.
So, we simplify notation.
Recall that, as our persistence diagrams are all $(M,m)$-bounded (i.e., in $\dgms$),
all points lie in $[-M,M]^2$.
For $i \in \Z_{\ge 0}$, we~define
\begin{itemize}
    \item $\latticeI{i} := \lattice{M}{2^{i+1}}$
    \item $\delta_i := 2M/2^i$, that is, the grid parameter for $\latticeI{i}$
    \item $\keyI{i}{P} := \key{P}{\latticeI{i}}$
    \item $\snapI{i}{P} := \snap{P}{\latticeI{i}}$
    \item $\delsnapI{i}{P} := \delsnap{P}{\latticeI{i}}$
\end{itemize}

\begin{definition}[Data Structure]\label{def:bigDS}
    Let $\Gamma \subset \dgms$ be finite,
    let~$c>\frac{3}{2}$,
    and let $\epsilon$ be the minimum
    bottleneck distance between any two diagrams in~$\Gamma$.
    Let $\tau =\bigdsUP$, then for
    each integer $i \in \{ 0, \ldots, \tau \}$,
    we define~$\Delta_i=\Delta_i(\Gamma)$ to be the data
    structure that stores the sorted list of keys $\{\keyI{i}{P_t}\}_{i,t,P}$, for
    each $P_t \in \delsnapI{i}{P}$ and $P \in \Gamma$.
    With each key, we store
    a list of distinct
    persistence diagrams from $\Gamma$ which have a snap-rounding to that
    key, and the number of distinct diagrams from $\Gamma$ which snap to the key.
    We note that a diagram with a given key can be found in time
    logarithmic in the number of distinct keys at that level.
    We denote the array of the multi-level data structure
    as~$\bigds=\bigds(\Gamma) := \{\Delta_i\}_{i=0}^{\tau}$. We can access a given
    level in constant time.
\end{definition}

In the definition above, the choice of $c$ and $\epsilon$ provides a point at which
the diagrams with the smallest bottleneck distance
stop colliding and we can stop considering smaller values of~$\delta$.
In particular, we choose $c>\frac{3}{2}$,
because the contrapositive of~\thmpartref{collision}{snapImpliesBound}
implies that if $d_B(P, Q) > \frac{3}{2}\delta_i$,
then~$\keyI{i}{P_*} \neq \keyI{i}{\widetilde{Q}}$ for any $P_* \in \delsnapI{i}{P}$.
Thus, we can guarantee that~$Q$ will share
a key with a representative of~$P$, for each~$Q$ close enough to~$P$.

\begin{remark}
    For each level,
    each diagram has $O(5^m)$ snap-roundings and keys
    that can each be generated in $O(m)$ time.
    Comparing two keys to determine their relative order requires~$O(m)$~time.

    For each diagram at each level,
    we can determine the set of unique keys in $O(m 5^m)$ by sorting
    the $O(5^m)$ keys and removing duplicates.
    Finding the unique keys for $n$ diagrams takes $O(nm5^m)$.
    Sorting the keys at a given level for $n$ diagrams takes
    $O(m(n5^m\log {n5^m})) = O(m(n5^m\log {n} + n5^m\log{5^m})) = O(m(n5^m\log {n} + n5^mm))
    = O(mn5^m(\log {n} + m))$ time.
    Creating a list of diagrams for each unique key at a given level requires $O(n5^m)$ time but
    this operation is asymptotically smaller than the complexity of sorting the
    keys.
    Then, generating the data structure $\bigds$ with $\tau$ levels
    takes $O(\tau (mn5^m(\log{n} + m))$ time.
\end{remark}

Next, we consider some properties of $\bigds$, specifically, that collisions on
a level of the data structure with a fine resolution imply collisions between
the same diagrams on levels with coarser resolutions.
To simplify notation, for $Q\in \dgms$ and $i \in \Z_{\geq 0}$, we let
$\Q_i$ be the diagram obtained from $Q$ by removing
all points less than or equal to distance~$\frac{1}{2}\delta_i$ of the
diagonal.
We again mention that all proofs can be found in the appendix.

\begin{lemma}[Hierarchical Collision]\label{lem:weakCollision}
    Let $\Gamma \subset \dgms$ be finite.
    Let $Q \in \dgms$ and $P \in \Gamma$.
    Let $j \in \Z_{\geq 0}$ and let
    $\Q_j$ and $\Q_i$ be the diagrams obtained from~$Q$ by removing
    all points less than or equal to distance~$\frac{1}{2}\delta_j$ (resp.,
    $\frac{1}{2}\delta_i$) of the
    diagonal.
    Suppose there exists $P_j \in \delsnapI{j}{P}$ such that~$\keyI{j}{P_j}=\keyI{j}{\Q_j}$
    (i.e.,~$P$ and $Q$ collide in level $\Delta_j$),
    then for any $i < j$, there exists $P_i \in \delsnapI{i}{P}$
    such that~$\keyI{i}{P_i}=\keyI{i}{\Q_i}$.
\end{lemma}

To find a near neighbor to $Q$ in $\Gamma$, we determine the last level
such that~$Q$ collides with an existing key in the data structure.
However, first we must consider where
the nearest neighbor lies relative to this level.

\begin{lemma}[Nearest Neighbor Bin]\label{lem:nnBin}
    Let $\Gamma, Q,$ and $\Q_i$ be as defined in \lemref{weakCollision}.
    Let $\Q_{i-2}$ be obtained from $Q$ by removing all points within
    $\frac{1}{2}\delta_{i-2}$ of the diagonal.
    Let~$P^{nn} \in \Gamma$ be the nearest neighbor of $Q$ in $\Gamma$, with respect to the
    bottleneck distance between diagrams.
    Let $i$ be the largest index such that $\keyI{i}{\Q_i}$ has a collision in
    $\Delta _i$.
    Then, there exists a snap-rounding~$P^{nn}_{i-2}$ in~$\delsnapI{i-2}{P^{nn}}$
    such that $\keyI{i-2}{P^{nn} _{i-2}} = \keyI{i-2}{\Q _{i-2}}$.
\end{lemma}

A result of \lemref{nnBin} is that if $i$ is the largest index such
that $Q$ has a collision
in $\Delta_i$, then we can construct examples in which $Q$ does not collide
with the nearest neighbor in $\Delta_i$, see
diagrams in \figref{nn-not-collide}.
\begin{figure}
\begin{center}
\includegraphics[width=.8\textwidth]{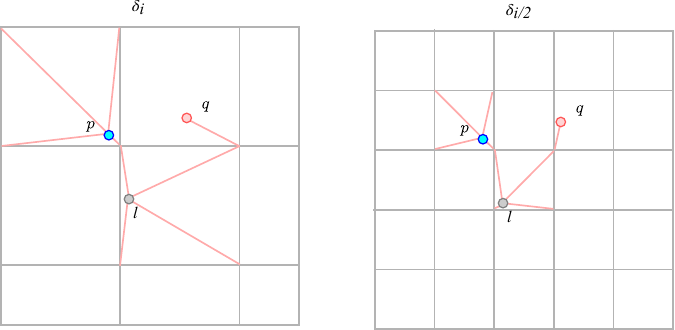}
\end{center}
\caption[Example for which nearest neighbor is not in last bin with collision]
{Situation in which the nearest neighbor of diagram $Q$
is not in the lowest bin that $Q$ has a collision in. Let query diagram
$Q$ be composed of the single point $q$ and let $P,L \in \dgms$, where
$P$ has one point $p$ and $L$ has one point $l$. We see that $d_B(P,Q)
= \frac{\delta_i}{2} + \epsilon_1$ for some small constant $\epsilon_1$
and $d_B(L,Q) = \frac{3\delta_i}{4}-\epsilon_2$ for some small
constant $\epsilon_2$. However, only $L$ and $Q$ collide at level
$\frac{\delta_i}{2}$, even though $d_B(P,Q) < d_B(L,Q)$.}
\label{fig:nn-not-collide}
\end{figure}

Next, we show that any diagram colliding with $Q$ in $\Delta_i$
is an approximate nearest neighbor.

\begin{lemma}[Nearest Neighbor Approximation]\label{lem:nn}
    Let $\Gamma, Q,$ and $\Q_i$ be as defined in \lemref{weakCollision}.
    Let $i$ be the largest index such that $\keyI{i}{\Q_i} \in \Delta_i$.
    Let~$P^{nn}\in \Gamma$ be the nearest neighbor of $Q$ in terms of bottleneck
    distance.
    The bottleneck distance between $Q$ and every diagram of $\Gamma$ with a key
    $\keyI{i}{\Q_i}$ in $\Delta_i$
    is a six-approximation of $d_B(P^{nn},Q)$.
\end{lemma}

The previous discussion implies that we can find an approximate
nearest neighbor by identifying the bin in the lowest level with a collision
and picking any diagram in that bin. Moreover, it
tells us that we can find the true nearest neighbor by
linearly searching for it through all diagrams with a collision two levels up.
Next, we prove that we can query for an approximate $k$th nearest neighbor for $k>1$.
First, we establish bounds on the location of $k$th nearest neighbor, generalizing
results from \lemref{nnBin}.

\begin{lemma}[$k$th-NN Location Upper Bound]
Let $\Gamma, Q,$ and $\Q_i$ be as defined in \lemref{weakCollision}.
Let~$P^k \in \Gamma$ be the $k$th nearest neighbor of $Q$ in $\Gamma$, with respect to the
bottleneck distance between diagrams.
Let $i$ be the largest index such that the number of distinct diagrams with snap-roundings
and keys equal to~$\keyI{i}{\Q_i}$ in $\Delta _i$ is at least~$k$.
Then, there exists a snap-rounding $P^k _{i-2}\in \delsnapI{i-2}{P^k}$ such that $\keyI{i-2}{P^k _{i-2}} = \keyI{i-2}{\Q _{i-2}}$.
\label{lem:knnLoc}
\end{lemma}

Then, we also bound the number of levels with a finer grid resolution that
the $k$th-nearest neighbor can collide with a snap-rounding of the query diagram.

\begin{lemma}[$k$th-NN Location Lower Bound]
Let $\Gamma, Q,$ and $\Q_i$ be as defined in \lemref{weakCollision}.
Let $P^k \in \Gamma$ be the $k$th nearest neighbor of $Q$ in $\Gamma$, with respect to the
bottleneck distance between diagrams.
Let $i$ be the largest index such that the number of distinct diagrams with snap-roundings
and keys equal to~$\keyI{i}{\Q_i}$ in $\Delta _i$ is at least~$k$.
Then, $P^k$ does not have a snap-rounding and key colliding with $\Q$ in
    any $\Delta _j$ such that $j > i+2$.
\label{lem:kNNDepth}
\end{lemma}

Using the previous two lemmas, we bound levels for which $P^k$ can
collide with~$\Q$.

\begin{corollary}[$k$th-NN Location]
Let $\Gamma, Q,$ and $\Q_i$ be as defined in \lemref{weakCollision}.
Let $P^k \in \Gamma$ be the $k$th nearest neighbor of $Q$, with respect to the
bottleneck distance between diagrams.
Let $i$ be the largest index such that the number of distinct diagrams with snap-roundings
and keys equal to~$\keyI{i}{\Q_i}$ at $\Delta _i$ is at least~$k$.
Let $j$ be the largest level in $\bigds$ such that there is a snap rounding and
key of $P^k$ colliding with $Q$.
Then, $i-2 \le j \le i+2$, i.e., $P^k$ must have a snap-rounding in, at most,
$\Delta _{i+2}$ and in, at least, $\Delta _{i-2}$.
\label{cor:kNNLevels}
\end{corollary}

To find the $k$th nearest neighbor to $Q$ in $\Gamma$,
we determine the last level of
$\bigds$ where~$Q$ has at least $k$ collisions. The proof is a
modification of the proof of \lemref{nn}.

\begin{lemma}[$k$th-Nearest Neighbor Approximation]
Let $\Gamma, Q,$ and $\Q_i$ be as defined in \lemref{weakCollision}.
Let $k$ be a positive integer greater than one.
Let $P^k \in \Gamma$ be the $k$th nearest neighbor of $Q$, with respect to the
bottleneck distance between diagrams.
Let $i$ be the largest index such that the number of distinct diagrams with snap-roundings
and keys equal to~$\keyI{i}{\Q_i}$ at $\Delta _i$ is at least~$k$.
The bottleneck distance between $Q$ and every diagram of $\Gamma$ with a key
$\keyI{i}{\Q_i}$ in $\Delta_i$
is a $24$-approximation of the $d_B(P^k,Q)$.
\label{lem:knnApprox}
\end{lemma}

\begin{remark}
The approximation factor is controlled by the last level
    where $P^k$ and $Q$ collide.
So, while we do not propose an efficient test for identifying the last level,
    we observe that in some cases,
    the approximation factor is much tighter.
For example,
    if $P^k$ last collides with~$Q$ in $\Delta_{i-2}$,
    then for all $P \in \Gamma$ that collide with $Q$ in $\Delta_i$,
    $d_B(P,Q) \le \frac{3}{2} d_B(P^k,Q)$.
\end{remark}

We now identify approximate nearest neighbors
for a query diagram.

\begin{theorem}[Approximate Nearest Neighbor Query]
Let $\Gamma, Q,$ and $\Q_i$ be as defined in \lemref{weakCollision}.
Let $n=|\Gamma|$ and let $\bigds$ be the multi-level structure described
in \defref{bigDS} with $\tau$ levels.
Then, the data structure $\bigds$ is of size  $\spcone$ and supports finding a
six-approximation of the nearest neighbor of $Q$ in $\Gamma$ in
    $\qtone$~time.
\label{thm:approxnn}
\end{theorem}

\begin{remark}
We note that exponential search could replace binary search for both finding
the last $\Delta _i$ where $\keyI{i}{\Q _i}$ collides with another key as well as
on each $\Delta _i \in \bigds$.
If $i$ is the largest~$\Delta _i$ such that the snap-rounding of
$\Q _i$ collides with another key, and~$\gamma$ is the index of the
key in $\Delta _i$ that collided with $\keyI{i}{\Q _i}$ then the query
time becomes $\qtexpsearch$.
\end{remark}

Finally, we prove that this data structure can provide responses to queries requesting
the $k$-nearest neighbors. Specifically, the $k$-nearest neighbors returned
are a $24$-approximation of the $k$th nearest neighbor.

\begin{corollary}[$k$-Nearest Neighbor Query]
Let $\Gamma, Q,$ and $\Q_i$ be as defined in \lemref{weakCollision}.
Let $n=|\Gamma|$ and let $\bigds$ be the multi-level structure described
in \defref{bigDS} with $\tau$ levels.
There exists a data structure of size $\spcone$
that supports finding $k$ diagrams that are each, in the worst case, a
$24$-approximation of the $k$th nearest neighbor of $Q$ from $\Gamma$ in
    \mbox{$\qtknn$ time}.
\label{cor:approxkNN}
\end{corollary}

\section*{Space Complexity Discussion}
While searching $\bigds$ is logarithmic in the number of diagrams, the data
structure becomes very large when the diagrams have even a moderate number of
points. For example, with~$m = 15$, we may have over $2^{34}$ keys at a given
level.
In this section, we discuss approaches intended to mitigate the size of the
data structure and explain why many become unrealistic in practice.

One way to reduce the size of each level is to ``flip'' the key generation and
querying.
In particular, instead of generating many keys for each diagram in
$\Gamma$, we compute one for each diagram. Then, for a query diagram~$Q$,
we compute and search~$O(5^m)$ keys at each level.
This may be practical for scenarios in which diagrams
in $\Gamma$ are large, but the query diagrams are small. Moreover, since
searching for a key is independent of the other keys, searching in parallel
follows naturally.
Furthermore, this approach reduces the size of $\bigds$ from $O(n5^m\tau)$
to $O(nm\tau)$.
However, this trades exponential space for exponential search time relative
to the number of points in the diagram.
Preliminary work on more complex snapping schemes have shown that we can
reduce the size of the data structure.  While the size is still exponential in
$m$, some of the more promising schemes have been able to reduce the expected
size to $O(2.6^m)$ and increase the approximation factor to a larger
constant value.
The next natural approach to pursue is probabilistic snap-rounding.

Consider generating a single key for each diagram in $\Gamma$ by snap-rounding
each diagram a single time at each grid resolution.
If we randomly snap-round our query diagram a polynomial number of times,
we may never collide with a diagram with small bottleneck distance since
only one of the~$O(5^m)$ snap-roundings may yield a match.
Another approach is defining a match to be when a certain ratio of points
collide.
However, we can run into situations where we return diagrams with
exceptionally large bottleneck distance.
For example, consider a class of diagrams with $k$ off-diagonal points
where for any two diagrams $\alpha$ and $\beta$ in this class,
$d_B(\alpha, \beta)$ is large, but if a subset of $k-1$ points from each
are considered, the bottleneck distance is small.
Then, if we query with a diagram $Q$ in which $k-1$ points from $Q$ have a small
bottleneck distance to $k-1$ points of
any diagram in the class we will return a match for every diagram.
However, the actual bottleneck distance between the query diagram and any diagram
in this class can be arbitrarily large.

Finally, while our space complexity is exponential in
diagram size, our queries
do not rely on probabilistic snap-roundings and the problems discussed
above resulting from these approaches.
We conclude this discussion
by offering an analysis of the complexity of our data structure,
in practice, relative to similar approaches used on polygonal curves
from \cite{DS17} and demonstrate a substantial decrease in size.
While we are comparing approaches for different problem domains, our
goal is to draw attention to the complexity related to the number
of points in the polygonal curves and diagrams to emphasize the
challenges related to reducing data structure size.
Specifically, for a data structure, storing $n$ curves in $\R^2$,
each with complexity at most~$m=15$, the constant
factor approximation from \cite{DS17}
requires $O(2^{120}n\log n + 15n)$ space; whereas, our
approach to storing $n$ diagrams with at most $15$ off-diagonal points
requires $O(n5^{15}\tau)$~space.
However, we note that our approach depends on the distribution
of points in the diagrams which dictates the size of $\tau$ while the
approach found in \cite{DS17} is independent of the distribution of curves.
\section{Concluding Remarks}
In this paper, we address the problem of supporting
approximate nearest neighbor search for a query persistence diagram among a finite
set $\Gamma$
of $(M,m)$-bounded persistence diagrams.
To the best of our knowledge,
this result is the first to introduce a method of searching a set
of persistence diagrams with a query diagram with performance guarantees
that does not require a linear number of bottleneck distance computations.
We utilize ideas from locality-sensitive hashing
along with a snap-rounding technique to generate keys
for a data structure which supports searching $\bigds:=\bigds(\Gamma)$
which has $\tau$ levels.
Specifically, when $|\Gamma|=n$,
the search time for an $(M,m)$-bounded query diagram is $\qtone$ and returns an
approximate nearest persistent diagram within a factor of six.
Additionally, searching for $k$ approximate
nearest neighbors can be done in $\qtknn$ and each of the diagrams are
within a factor of twenty-four of the $k$th nearest neighbor.

For simplicity, we assumed that none of the points in the diagrams of $\Gamma$
are on grid lines.  To handle points on grid lines, we add additional keys.
More specifically, for a diagram~$P \in \Gamma$ and a grid $\latticeI{i}$ in
which a point $p \in P$ is on a grid point in~$\latticeI{i}$,
we snap $p$ to its nine nearest
grid points. If $p$ is on a grid edge of $\latticeI{i}$ (and not
a grid point) we snap to~$p$ to its six nearest grid points.  While the additional
keys increase the size of $\delsnapI{i}{P}$, the space complexity of storing or
time complexity of querying $\bigds$ does not change.

This paper is just one of the first steps towards practical searches in
the space of persistence diagrams.
Future work consists of an implementation of the data structure,
exploring additional representations of the persistence diagrams
to employ more techinues from locality-sensitive hashing to
reduce space and time complexity, and using techniques
in this paper to expand the results in \cite{DS17}.

\begin{acknowledgements}
    We thank Michael Kerber for his discussions about early drafts
    of this paper and future work,
    Donald R. Sheehy for directing us to additional searching approaches and
    feedback on early drafts,
    and Brendan Mumey for the discussions related to this
    research and pointing out
    the paper by Heffernan and Schirra \cite{heffernan1994approximate}.
\end{acknowledgements}

%
%

\bibliographystyle{spmpsci}      
\bibliography{references}

\end{document}